\documentclass[a4paper,UKenglish,cleveref, autoref, thm-restate]{lipics-v2021}
\title{Feedback Vertex Set and Even Cycle Transversal for $H$-Free Graphs: Finding Large Block Graphs\thanks{An extended abstract containing some of the results in this paper appeared in the proceedings of MFCS 2021~\cite{PPR21}.}}
\titlerunning{Feedback Vertex Set and Even Cycle Transversal for $H$-Free Graphs}

\author{Giacomo Paesani}{Department of Computer Science, Durham University, UK}{giacomo.paesani@durham.ac.uk}{https://orcid.org/0000-0002-2383-1339}{}

\author{Dani\"el Paulusma}{Department of Computer Science, Durham University, UK}{daniel.paulusma@durham.ac.uk}{https://orcid.org/0000-0001-5945-9287}{Supported by the Leverhulme Trust (RPG-2016-258).}

\author{Pawe{\l} Rz{\k a}\.{z}ewski}{Faculty of Mathematics and Information Science, Warsaw University of Technology, Poland \and Faculty of Mathematics, Informatics, and Mechanics, University of Warsaw, Poland}{p.rzazewski@mini.pw.edu.pl}{https://orcid.org/0000-0001-7696-3848}{Supported by Polish National Science Centre grant no. 2018/31/D/ST6/00062.}

\authorrunning{G. Paesani, D. Paulusma and P. Rz{\k a}\.{z}ewski}

\Copyright{Giacomo Paesani, Dani\"el Paulusma and Pawe{\l} Rz{\k a}\.{z}ewski}



\nolinenumbers 

\hideLIPIcs

\usepackage{amssymb,amsmath,graphicx}
\usepackage{mathtools}
\usepackage{pgf}
\usepackage{color,enumerate,comment,boxedminipage}
\usepackage{float}
\usepackage{hyperref}
\usepackage{comment}
\usepackage{tikz}
\usepackage{xspace}
\usepackage{subcaption}
\usepackage{algorithm}
\usepackage[noend]{algpseudocode}
\usepackage[T1]{fontenc}
\usetikzlibrary{shapes,snakes}

\newcommand{\wei}{\mathfrak{w}}
\newcommand{\Q}{\mathbb{Q}}
\newcommand{\cI}{\mathcal{I}}
\newcommand{\cC}{\mathcal{C}}
\newcommand{\cB}{\mathcal{B}}
\newcommand{\cX}{\mathcal{X}}
\newcommand{\cS}{\mathcal{S}}
\newcommand{\Oh}{\mathcal{O}}
\newcommand{\CBG}{{\sc Max $\cC$-Block Graph}\xspace}
\newcommand{\Blob}[1]{#1^\circ}
\newcommand{\BlockTree}[1]{\mathsf{BCF}(#1)}
\newcommand{\Blocks}[1]{\mathsf{Blocks}(#1)}
\newcommand{\Cuts}[1]{\mathsf{Cutvertices}(#1)}

\newcommand{\NP}{{\sf NP}}

\newcommand{\ssi}{\subseteq_i}
\newcommand{\si}{\supseteq_i}

\newcommand{\problemdef}[3]{
	\begin{center}
		\begin{boxedminipage}{.99\textwidth}
			\textsc{{#1}}\\[2pt]
			\begin{tabular}{ r p{0.8\textwidth}}
				\textit{~~~~Instance:} & {#2}\\
				\textit{Objective:} & {#3}
			\end{tabular}
		\end{boxedminipage}
	\end{center}
}

\usepackage[backgroundcolor = blue!50]{todonotes}

\newtheorem{claimm}{Claim}

\begin{document}

\maketitle

\begin{abstract}
We prove new complexity results for {\sc Feedback Vertex Set} and {\sc Even Cycle Transversal} on $H$-free graphs, that is, graphs that do not contain some fixed graph $H$ as an induced subgraph. In particular, we prove that for every $s\geq 1$, both problems are polynomial-time solvable for $sP_3$-free graphs and $(sP_1+P_5)$-free graphs;
here, the graph $sP_3$ denotes the disjoint union of $s$ paths on three vertices and the graph $sP_1+P_5$ denotes the disjoint union of $s$ isolated vertices and a path on five vertices.
Our new results for {\sc Feedback Vertex Set} extend all known polynomial-time results for {\sc Feedback Vertex Set} on $H$-free graphs, namely for $sP_2$-free graphs [Chiarelli et al., TCS 2018], $(sP_1+P_3)$-free graphs [Dabrowski et al., Algorithmica 2020] and $P_5$-free graphs [Abrishami et al., SODA 2021].
Together, the new results also show that both problems exhibit the same behaviour on $H$-free graphs (subject to some open cases). This is in part due to a new general algorithm we design for finding in a ($sP_3)$-free or $(sP_1+P_5)$-free graph $G$ a largest induced subgraph whose blocks belong to some finite class ${\cal C}$ of graphs. We also compare our results with the state-of-the-art results for the {\sc Odd Cycle Transversal} problem, which is known to behave differently on $H$-free graphs.
\end{abstract}

\section{Introduction}\label{s-intro}
For a set of graphs ${\cal F}$, an {\it ${\cal F}$-transversal} of a graph $G$ is a set of vertices that intersects the vertex set of every (not necessarily induced) subgraph of $G$ that is isomorphic to some graph of~${\cal F}$.
The problem {\sc Min ${\cal F}$-Transversal} (also called ${\cal F}$-{\sc Deletion}) is to find an ${\cal F}$-transversal of minimum size (or size at most $k$, in the decision variant).
Graph transversals form a central topic in Discrete Mathematics and Theoretical Computer Science, both from a structural and an algorithmic point of view. 

If ${\cal F}$ is the set of all cycles, the set of all even cycles or odd cycles, then we obtain the problems {\sc Feedback Vertex Set}, {\sc Even Cycle Transversal} and {\sc Odd Cycle Transversal}, respectively. All three problems are \NP-complete; hence, they have been studied for special graph classes, in particular {\it hereditary} graph classes, that is, classes closed under vertex deletion. Such  classes can be characterized by a (unique) set ${\cal H}$ of minimal forbidden induced subgraphs. Then, in order to initiate a systematic study, it is standard to first consider the case where ${\cal H}$ has size~$1$, say ${\cal H}=\{H\}$ for some graph~$H$.

We aim to extend known complexity results for {\sc Feedback Vertex Set} for $H$-free graphs and to perform a new, similar study for {\sc Even Cycle Transversal} (for which, so far, mainly parameterized complexity results exist~\cite{AGHKKKKO20,BBBK20,KLPS17,MRRS12}). To describe the known and new results we need some terminology. 
The cycle and path on $r$ vertices are denoted $C_r$ and $P_r$, respectively.
The \emph{disjoint union} of two vertex-disjoint graphs~$G_1$ and~$G_2$ is the graph $G_1+G_2= (V(G_1)\cup V(G_2), E(G_1)\cup E(G_2))$. We write $sG$ for the disjoint union of $s$ copies of $G$.  For a set $S\subseteq V$, let~$G[S]$ be the subgraph~of~$G$ induced by~$S$.
We write $H\ssi G$ (or $G\si H$) if $H$ is an induced subgraph~of~$G$.

\subsection{Known Results}\label{s-known}
 
By Poljak's construction~\cite{Po74}, for every integer~$g\geq 3$, {\sc Feedback Vertex Set}  is \NP-complete for graphs of girth at least~$g$ (the {\it girth} of a graph is the length of its shortest cycle).  
The same holds for {\sc Odd Cycle Transversal}~\cite{CHJMP18}.  It is also known that {\sc Feedback Vertex Set}~\cite{Mu17b} and {\sc Odd Cycle Transversal}~\cite{CHJMP18} are
\NP-complete for line graphs and thus for claw-free graphs (the claw is the $4$-vertex star).  Hence, both problems are \NP-complete for the class of $H$-free graphs whenever $H$ has a cycle or claw. A graph with no cycles and no claws is a forest of maximum degree at most~$2$.  Thus, it remains to consider the case where $H$ is a {\it linear forest}, that is, a collection of disjoint paths.  Both problems are polynomial-time solvable
on permutation graphs~\cite{BK85} and thus on $P_4$-free graphs~\cite{BK85},
on  $sP_2$-free graphs for every $s\geq 1$~\cite{CHJMP18} and on $(sP_1+P_3)$-free graphs for every $s\geq 0$~\cite{DFJPPR19}. Additionally, {\sc Feedback Vertex Set} is polynomial-time solvable on $P_5$-free graphs~\cite{ACPRS21}, and
{\sc Odd Cycle Transversal}  is \NP-complete for $(P_2+P_5,P_6)$-free graphs~\cite{DFJPPR19}.
A similar \NP-hardness result for  {\sc Feedback Vertex Set} or {\sc Even Cycle Transversal} is unlikely:
 for every linear forest~$H$,
both problems are quasipolynomial-time solvable on $H$-free graphs~\cite{GLPPR21} 
(see Section~\ref{s-con} for details).

\subsection{New Polynomial-Time Results}\label{s-our}

We first note that {\sc Min ${\cal F}$-Transversal} is polynomially equivalent to {\sc Max Induced ${\cal F}$-Subgraph}, the problem of finding a maximum-size induced subgraph of the input graph $G$ that does not belong to ${\cal F}$ (where we assume that $G$ has at least one such subgraph).
We say that {\sc Max Induced ${\cal F}$-Subgraph} is the {\it complementary} problem of {\sc Min ${\cal F}$-Transversal}, and vice versa.
For example, setting ${\cal F}=\{P_2\}$ yields the well-known complementary problems {\sc Min Vertex Cover} and {\sc Max Independent Set}.

Using the complementary perspective, we now argue that {\sc Feedback Vertex Set} and {\sc Even Cycle Transversal} are closely related, in contrast to {\sc Odd Cycle Transversal}.
A graph $G$ is \emph{biconnected} if it has at least two vertices, is connected, and $G-u$ is connected for every $u\in V(G)$.
A \emph{block} of a graph $G$ is an inclusion-wise maximal biconnected subgraph of $G$.
We now let $\cC$ be a set of biconnected graphs. A graph~$G$ is a {\it $\cC$-block graph} if every block of~$G$ is isomorphic to some graph in $\cC$.
If $\cC=\{P_2\}$, then $\cC$-block graphs 
are precisely forests, and if $\cC=\{P_2,C_3,C_5,C_7,\ldots\}$, then $\cC$-block graphs are  
called {\it odd cacti}. It is well known that a graph is an odd cactus if and only if it does not contain any even cycle as a subgraph.
Hence, the complementary problems of {\sc Even Cycle Transversal} and {\sc Feedback Vertex Set} are somewhat similar:
in particular, both forests and odd cacti have bounded treewidth and their blocks have a very simple structure.
This is in stark contrast to {\sc Odd Cycle Transversal}, whose complementary problem is to find a large induced bipartite subgraph, which might be arbitrarily complicated.

The commonality of complementary problems of {\sc Even Cycle Transversal} and {\sc Feedback Vertex Set} leads to the following optimization problem, where ${\cal C}$ is some fixed class of biconnected graphs, that is, ${\cal C}$ is not part of the input but specified in advance.
Note that we consider the more general setting in which every vertex $v$ of $G$ is equipped with a weight $\wei(v) > 0$, and we must find a solution with maximum total weight.

\problemdef{\CBG}{a graph $G=(V,E)$ with a vertex weight function $\wei : V \to \Q^+$.}{find a maximum-weight set $X \subseteq V$ such that $G[X]$ is a $\cC$-block graph.}

We observe that \CBG\ is well-defined for every set $\cC$, including $\cC = \emptyset$, as every independent set in a graph forms a solution.
A restriction of the \CBG problem was introduced and studied from a parameterized complexity perspective by 
Bonnet et al.~\cite{BBKM16} as {\sc Bounded ${\cal C}$-Block Vertex Deletion} (so from the complementary perspective) 
where each block must in addition have bounded size.

In Section~\ref{s-pre} we slightly extend a previously known result, concerning the so-called \emph{blob graphs}~\cite{GLPPR21}.
This extended version of the result forms a key ingredient for the proof of our main results, shown in Sections~\ref{s-poly} and~\ref{s-mmm}, respectively, which are the following two theorems.

\begin{theorem}\label{thm:main}
For every integer $s\geq 1$ and every \emph{finite} class $\cC$ of biconnected graphs, 
\CBG 
can be solved in polynomial time for $sP_3$-free graphs.
\end{theorem}

\begin{theorem}\label{thm:P5sP1}
For every integer $s\geq 1$ and every \emph{finite} class $\cC$ of biconnected graphs, 
\CBG 
can be solved in polynomial time for $(sP_1+P_5)$-free graphs.
\end{theorem}

\noindent
We note that $sP_3$-free graphs are the graphs that become a disjoint union of cliques after removing the vertices of any induced $(s-1)P_3$ and their neighbours. The class of $(sP_1+P_5)$-free graphs is a natural generalization of the class of $(P_1+P_5)$-free graphs. The latter graphs are also known as the  {\it nearly $P_5$-free graphs}, that is, graphs in which the subgraph induced by the non-neighbourhood of 
any vertex is $P_5$-free. More generally, a graph is {\it nearly} $\pi$ for some graph property~$\pi$ if the subgraph induced by the non-neighbourhood of 
any vertex has property~$\pi$. It is easy to see that {\sc Max Independent Set} is polynomial-time solvable for graphs that are nearly~$\pi$ if it is so for graphs with property~$\pi$ (see, for example,~\cite{BH07}). However, for many other graph problems, including the problems we study in this paper, such a statement either does not hold, is not known, or could be non-trivial to prove even for graphs that are nearly $P_5$-free (such as for example, {\sc Connected Vertex Cover}~\cite{JPP20}).

We prove both theorems using the same technique. Essentially we reduce to {\sc Max Independent Set} for $sP_3$-free blob graphs and $(sP_1+P_5)$-free blob graphs, respectively. In order to do this, we first perform a structural analysis of $sP_3$-free $\cC$-block graphs and $(sP_1+P_5)$-free $\cC$-block graphs.
These analyses have some common elements, namely they are based on the so-called block-cut forest of the (unknown) maximum-weight solution $F$. This forest contains as its vertices the cutvertices $x$ and blocks $b$ of $F$ such that $xb$ is an edge if and only if $x$ belongs to $b$. The precise arguments are different and the resulting polynomial-time algorithms exploit the $sP_3$-freeness and $(sP_1+P_5)$-freeness in different ways.

\subsection{Implications and New Hardness Results}

Theorems~\ref{thm:main} and~\ref{thm:P5sP1} imply corresponding results for {\sc Feedback Vertex Set}, as the latter problem is equivalent to 
\textsc{Max $\{P_2\}$-Block Graph}.
The condition for ${\cal C}$ to be finite is critical for our proof technique.
Nevertheless, we still have the corresponding result for {\sc Even Cycle Transversal} as well: for $sP_3$-free graphs, the cases $\cC=\{P_2,C_3,C_5,C_7,\ldots\}$ and $\cC=\{P_2,C_3,C_5,\ldots,C_{4s-3}\}$ are equivalent. Note that we cannot make such an argument for {\sc Odd Cycle Transversal}, as arbitrarily large bicliques are $2P_3$-free.

\begin{corollary}\label{c-1}
For every integer $s\geq 1$, {\sc Feedback Vertex Set} and {\sc Even Cycle Transversal} can be solved in polynomial time for $sP_3$-free graphs and $(sP_1+P_5)$-free graphs.
\end{corollary}

\noindent
Corollary~\ref{c-1} extends the aforementioned results for {\sc Feedback Vertex Set} on $sP_2$-free graphs and $(sP_1+P_3)$-free graphs. In Section~\ref{s-ect} we prove that {\sc Even Cycle Transversal} is \NP-complete for graphs of large girth and for line graphs, and consequently, for $H$-free graphs where $H$ contains a cycle or a claw.
Hence, {\sc Feedback Vertex Set} and {\sc Even Cycle Transversal} behave similarly on $H$-free graphs, subject to a number of open cases, which we listed in Table~\ref{t-table}.

\newcommand\Tstrut{\rule{0pt}{2.6ex}}        
\newcommand\Bstrut{\rule[-0.9ex]{0pt}{0pt}}
\begin{table}[h]
\begin{tabularx}{\textwidth}{ p{1.1cm} p{3.45cm} X  p{2.6cm}  }
    \hline
    & polynomial-time 
    & unresolved & $\NP$-complete\Tstrut\Bstrut \\
  \hline
    {\sc FVS} &  \textcolor{blue}{$H \ssi sP_1+P_5$ or} \newline \phantom{$H\ssi$} \textcolor{blue}{$sP_3$ for $s\geq 1$}  & $H\si P_2+P_4$ or $P_6$ & none\Tstrut\Bstrut \\  \hline
    {\sc ECT} &  \textcolor{blue}{$H \ssi sP_1+P_5$ or} \newline \phantom{$H\ssi$} \textcolor{blue}{$sP_3$ for $s\geq 1$}  & $H\si P_2+P_4$ or $P_6$ & none\Tstrut\Bstrut \\  \hline
   {\sc OCT} & $H=P_4$ or \newline  $H\ssi sP_1+P_3$ or \newline \phantom{$H\ssi$} $sP_2$ for $s\geq 1$ &    
    $H=sP_1+P_5$ for $s\geq 0$  or \newline $H=sP_1+tP_2+uP_3+vP_4$ \newline for $s,t,u \ge 0$, $v\geq 1$
    \newline
    with $\min\{s,t,u\}\geq 1$ if $v=1$, or
    \newline
    $H=sP_1+tP_2+uP_3$ for $s,t\geq 0$, $u\geq 1$
    with $u\geq 2$ if $t=0$
     & $H\si P_6$ or $P_2+P_5$\Tstrut\Bstrut \\  \hline \\[-0.25cm]
\end{tabularx}
 \caption{The complexity of {\sc Feedback Vertex Set} (FVS), {\sc Even Cycle Transversal} (ECT)
 and {\sc Odd Cycle Transversal} (OCT) on $H$-free graphs for a linear forest $H$.  All three problems are $\NP$-complete for $H$-free graphs when $H$ is not a linear forest (see also Section~\ref{s-ect}). The blue cases (for FVS and ECT) are the {\it algorithmic} contributions of this paper. We write $H\ssi H'$ if $H$ is an induced subgraph of $H'$. 
See Section~\ref{s-known} for references to the known results in the table.}\label{t-table}
\end{table}

\section{Blob Graph of Graphs With No Large Linear Forest}\label{s-pre}

Let $G=(V,E)$ be a graph. A {\it  (connected) component}  is a maximal connected subgraph of~$G$.
The \emph{neighbourhood} of a vertex $u\in V$ is the set $N_G(u)=\{v\; |\; uv\in E\}$. For $U\subseteq V$, we let $N_G(U)=\bigcup_{u\in U}N(u)\setminus U$. Two sets $X_1,X_2 \subseteq V(G)$ are \emph{adjacent} if $X_1 \cap X_2 \neq \emptyset$ or there exists an edge with one endvertex in $X_1$ and the other in $X_2$. 
The \emph{blob graph} $\Blob{G}$ of $G$ is defined as follows.
\[
V(\Blob{G}) \coloneqq \{ X \subseteq V(G) ~|~ G[X] \text{ is connected} \} \text{ and } E(\Blob{G}) \coloneqq \{ X_1X_2 ~|~ X_1 \text{ and } X_2 \text{ are adjacent} \}.
\]
Gartland et al.~\cite{GLPPR21} showed that for every graph $G$, the length of a longest induced path in~$\Blob{G}$ is equal to the length of a longest induced path in $G$. We slightly generalize this result.

\begin{theorem} \label{thm:blob}
For every linear forest $H$, a graph $G$ contains $H$ as an induced subgraph if and only if $\Blob{G}$ contains $H$ as an induced subgraph.
\end{theorem}
\begin{proof}
As $G$ is an induced subgraph of $\Blob{G}$, the $(\Rightarrow)$ implication is immediate. We prove the $(\Leftarrow)$ implication
by induction on the number $k$ of connected components of $H$.
If $k=1$, then the claim follows directly from the aforementioned result of Gartland et al.~\cite{GLPPR21}.
So assume that $k \geq 2$ and the statement holds for all linear forests $H$ with fewer than $k$ connected components.
Let $P'$ be one of the connected components of $H$, and define $H' \coloneqq H - P'$.

Suppose that $\Blob{G}$ contains an induced subgraph 
isomorphic to $H$.
Let $\mathcal{X}$ be the set of vertices of $\Blob{G}$, such that $\Blob{G}[\mathcal{X}]$ is isomorphic to $H$.
Furthermore, let $\mathcal{Y} \subseteq \mathcal{X}$ be the set of vertices that induce in $\Blob{G}[\mathcal{X}]$ the component $P'$ of $H$,
that is, $\Blob{G}[\mathcal{Y}]$ is isomorphic to $P'$.

Let $Y \subseteq V(G)$ be the union of sets in $\mathcal{Y}$.
Note that $\Blob{G}[\mathcal{Y}]$ is an induced subgraph of $\Blob{(G[Y])}$.
Thus, by the inductive assumption, $G[Y]$ contains an induced copy of $P'$.

Let $X \subseteq V(G)$ be the union of sets in $\mathcal{X} \setminus \mathcal{Y}$.
Since the copy of $H$ in $\Blob{G}$ is induced, we know that in $\Blob{G}$ there are no edges between $\mathcal{X} \setminus \mathcal{Y}$ and $\mathcal{Y}$. This is equivalent to saying that $X \cap N[Y] = \emptyset$.
So we conclude that $\Blob{G}[\mathcal{X} \setminus \mathcal{Y}]$ is an induced subgraph of $\Blob{(G-N[Y])}$.
Since $\Blob{G}[\mathcal{X} \setminus \mathcal{Y}]$, and thus $\Blob{(G-N[Y])}$, contains an induced copy of $H'$,
by the inductive assumption we know that $G-N[Y]$ contains an induced copy of $H'$.
Combining this subgraph with the induced copy of $P'$ in $G[Y]$, we obtain an induced copy of $H$ in $G$.
\end{proof}

\section{The Proof of Theorem~\ref{thm:main}}\label{s-poly}

We start with analyzing the structure of $sP_3$-free $\cC$-block graphs in Section~\ref{s-str}, where $\cC$ is any finite class of biconnected graphs. Then, in Section~\ref{s-alg}, we  present our algorithm for \CBG on $sP_3$-free graphs.
 
\subsection{Structural Lemmas}\label{s-str}

From now on, let $\cC$ be a finite class of biconnected graphs.
For some fixed positive integer~$s$, let $G=(V,E)$ be an $sP_3$-free graph with $n$ vertices and vertex weights $\wei \colon V \to \Q^+$.
Let $X \subseteq V$ such that $F=G[X]$ is a $\cC$-block graph.
A component of $F$ is \emph{trivial} if it is a single vertex or a single block, otherwise it is \emph{non-trivial}. Let $F'$ be the graph obtained from $F$ by removing all trivial components. Note that $F'$ and $F$ are $sP_3$-free, as $G$ is $sP_3$-free.

We denote the set of cutvertices of $F'$ and the set of blocks of $F'$ by $\Cuts{F'}$ and $\Blocks{F'}$, respectively.
The \emph{block-cut forest} $\BlockTree{F'}$ of $F'$ has vertex set 
$\Cuts{F'} \cup \Blocks{F'}$ 
and an edge set that consists of all edges $xb$ such that $x \in \Cuts{F'}$ and $b \in \Blocks{F'}$, and $x$ belongs to $b$. By definition, each component of $F'$ has a cutvertex; we pick an arbitrary one as root for the corresponding tree in $\BlockTree{F'}$ to get a parent-child relation. 
Each leaf of $\BlockTree{F'}$ belongs to $\Blocks{F'}$, and we call such blocks \emph{leaf blocks}.

A cutvertex $x$ of $F'$ is a \emph{terminal of type 1} if $x$ has at least two children in $\BlockTree{F'}$ that are leaves, whereas
$x$ is a \emph{terminal of type~2} if there exists a leaf block, whose great-grandparent in $\BlockTree{F'}$ is $x$. In the latter case, there is a three-edge downward path from $x$ to a leaf in $\BlockTree{F'}$; see also Fig.~\ref{fig:terminals}. Let $d$ be the maximum number of vertices of a graph in $\cC$.

\begin{figure}[t]
\centering
\includegraphics[scale=1.1,page=1]{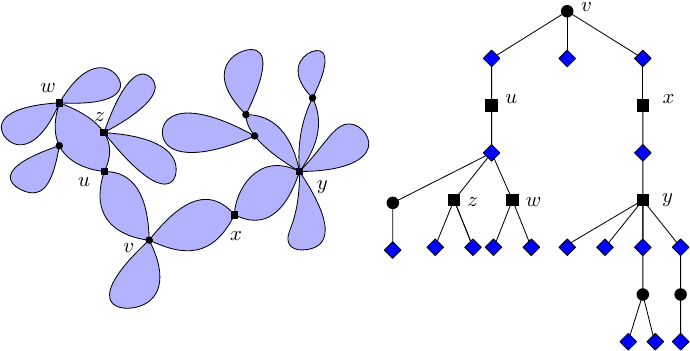}
\caption{Left: a graph $F'$. Blue shapes are blocks, squares are terminals, and dots are non-terminal cutvertices.
Right: $\BlockTree{F'}$, rooted in the cutvertex $v$. Blue diamonds are blocks; $w$
and $z$ 
are terminals of type~1, $u$ and $x$ are terminals of type~2, and $y$ is a terminal of both types.
The remaining cutvertices are not terminals. {\it We will also use this example with this particular $\BlockTree{F'}$ in later figures.}}
\label{fig:terminals} 
\end{figure}

\begin{lemma}\label{t-1}
At most $d \cdot (s-1)$ vertices of $F'$ are terminals of type~1.
\end{lemma}

\begin{proof}
For contradiction, suppose that there are at least $d \cdot (s-1)+1$ terminals of type~1.
We observe that $F'$ is $d$-colourable. Indeed, each block has at most $d$ vertices, so $d$ colours are sufficient to colour each block.
Furthermore, we can permute the colours in each block, so that the colourings agree on cutvertices.

This implies that there is an independent set $X$ of size at least $s$, whose every element is a terminal of type~1.
For each such terminal $v$, let its \emph{private $P_3$} be a 3-vertex path with $v$ as the central vertex and each endpoint belonging to a different leaf block that is a child of $v$ in $\BlockTree{F'}$.
Note that each private $P_3$ is induced. Furthermore, the private $P_3$'s of vertices in $X$ are pairwise non-adjacent:
this follows from the definition of terminals of type~1 and the fact that $X$ is independent.
Thus we have found an induced $sP_3$ in $F$, a contradiction.
\end{proof}

\begin{lemma}\label{t-2}
At most $(d+1) \cdot (s-1)$ vertices of $F'$ are terminals of type~2.
\end{lemma}

\begin{proof}
For contradiction, suppose that there are at least $(d+1) \cdot (s-1)+1$ terminals of type~2.
Observe that $F'$ has a proper $(d+1)$-colouring $f$, satisfying the following two properties:\\[-10pt]
\begin{enumerate}
\item the vertices in each block receive pairwise distinct colours and
\item if $b$ is a block, then any vertex of $b$ receives a colour which is different than the colour of the cutvertex which is the great-grandparent of $b$ in $\BlockTree{F'}$ (if such a cutvertex exists).\\[-10pt]
\end{enumerate}
\noindent
It is easy to find such a colouring of each tree in $\BlockTree{F'}$ by choosing an arbitrary colour for the root and proceeding in a top-down fashion. Suppose we want to colour the block $b$ and its parent in $\BlockTree{F'}$ is the cutvertex $v$.
Recall that $b$ has at most $d$ vertices and exactly one of them is already coloured.
Furthermore, we want to avoid the colour of the grandparent of~$v$ (if such a vertex exists), so we have sufficiently many free colours to colour each vertex of $b\setminus \{v\}$ with a different one.

Now, by our assumption, there is a set $X$ of at least $s$ terminals of type~2 that received the same colour in $f$.
For each $v \in X$, we define its \emph{private $P_3$} as follows.
Recall that by the definition of a terminal of type~2, there is a leaf block $b$, whose great-grandparent in $\BlockTree{F'}$ is $v$.
The private $P_3$ of $v$ is given by the first three vertices on a shortest path $P$ from $v$ to $b$.
Note that in the extreme case it might happen that both $b$ and its grandparent in $\BlockTree{F'}$ are edges,
but $P$ always has at least three vertices.

Clearly, each private $P_3$ is an induced path.
We claim that the private $P_3$'s associated with 
any two vertices of $X$ form an induced $2P_3$.
For contradiction, suppose otherwise. Let $v_1,v_2$ be distinct vertices of $X$,
and let $v_i,x_i,y_i$ be the consecutive vertices of the private $P_3$ associated with $v_i$.
Let $b_i$ be the block containing $v_i$ and $x_i$.

First, observe that the sets $\{v_1,x_1,y_1\}$ and $\{v_2,x_2,y_2\}$ are disjoint.
Indeed, we know that $v_1 \neq v_2$ 
by assumption, and because $\BlockTree{F'}$ is a rooted tree, we have that $\{x_1,y_1\} \cap \{x_2,y_2\} = \emptyset$.
Furthermore, recall that $f(v_1) = f(v_2)$ and by the definition of $f$, we have that the colours of $x_i$ and of $y_i$ are different from the colour of $v_i$.

So now suppose that there is an edge with one endvertex in $\{v_1,x_1,y_1\}$ and the other in $\{v_2,x_2,y_2\}$.
Clearly this edge cannot join $v_1$ and $v_2$, as the colouring $f$ is proper.
Furthermore, there is no edge between $\{x_1,y_1\}$ and $\{x_2,y_2\}$, as $v_1$ and $v_2$ are cutvertices of a rooted tree.
Suppose that $v_2$ is adjacent to $x_1$ (the case that $v_1$ is adjacent to $x_2$ is symmetric).
As each vertex of $b_1$ gets assigned a different colour by $f$
and $f(v_1)=f(v_2)$,
we observe that $v_2$ cannot belong to $b_1$.
Thus $x_1$ is a cutvertex. However, by the second property of $f$, we obtain that the colour of $v_2$ must be different from the colour of $v_1$. 
As $f(v_1)=f(v_2)$, this is a contradiction.

So finally suppose that $v_2$ is adjacent to $y_1$
(the case that $v_1$ is adjacent to $y_2$ is symmetric).
Note that then $y_1$ cannot belong to a leaf block, meaning that $y_1$ belongs to $b_1$.
Similarly to the previous paragraph, the definition of $f$ implies that the colour of $v_2$ must be different from the colour of $v_1$, a contradiction.

We conclude that $\{v_1,x_1,y_1,v_2,x_2,y_2\}$ induces a $2P_3$. As $v_1$ and $v_2$ were arbitrary vertices of $X$ and $|X|\geq s$,
this means we have found an induced $sP_3$ in $F'$, a contradiction.
\end{proof}

Lemmas~\ref{t-1} and~\ref{t-2} imply the following.

\begin{lemma}\label{t-3}
The number of terminals of $F'$ is at most $(2d+1) \cdot (s-1)$.
\end{lemma}

\begin{figure}[t]
\centering
\includegraphics[scale=1.1,page=5]{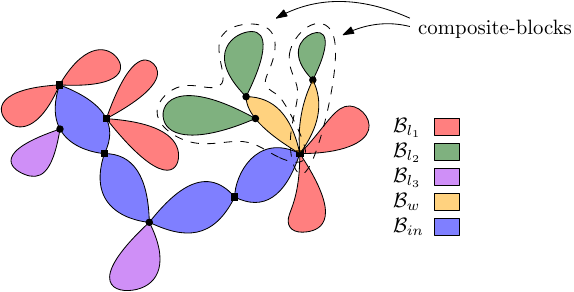}
\caption{The classification of blocks of the example of Figure~\ref{fig:terminals} (so this classification is based on the block-cut forest $\BlockTree{F'}$ from Fig.~\ref{fig:terminals}).}
\label{fig:blocktypes} 
\end{figure}

If $v$ is a terminal of type~$2$, then by definition there is a cutvertex $w$ that belongs to both a block containing $v$ as well as to some leaf block. We call such $w$ a \emph{witness} of $v$. 
We note that a (non-leaf) block may contain multiple witnesses.
Some of these witnesses might be terminals (and in that case they are of type~1), whereas other might not be terminals. 

We now partition the set of blocks of $F'$ into the following subsets; see also Fig.~\ref{fig:blocktypes}:\\[-10pt]
\begin{itemize}
\item $\cB_{l_1}$ is the set of leaf blocks 
containing a terminal of type~1,
\item $\cB_{l_2}$ is the set of leaf blocks containing a witness $w$ 
that is not a terminal of type~1,
\item $\cB_{l_3}$ is the set of remaining leaf blocks, that is, the ones with a cutvertex that is neither a terminal nor a witness,
\item $\cB_{w}$ is the set of blocks with at least two cutvertices, one of which is a terminal of type~2 and all the other ones are non-terminal witnesses of that type-2 terminal,
\item $\cB_{in}$ is the set of all remaining blocks.\\[-10pt]
\end{itemize}
Note that blocks in $\cB_{l_2}$ and $\cB_{w}$ come in groups, that is, for each block $B$ in $\cB_{w}$, there are at most $d-1$ blocks $B'$ in $\cB_{l_2}$, such that $B$ and $B'$ share a vertex,
(note that this common vertex is a non-terminal witness). 
We will consider
this group of
blocks as one object. 
Formally, a \emph{composite-block} is a graph $G[V(B)\cup \bigcup_{1\leq i\leq r} V(B_i)]$, for some integer $r\leq d-1$ chosen to be maximum, such that 

\begin{itemize}
\item $B \in \cB_{w}$, so $B$ has some vertex $u$ that is a terminal of type~$2$,
\item for every $i\in \{1,\ldots,r\}$, $B_i \in \cB_{l_2}$, and
\item for every $i\in \{1,\ldots,r\}$, $|V(B) \cap V(B_i)|=1$ and $u\notin V(B_i)$.
\end{itemize}

\noindent
Note that each composite-block has at most $(d-1)d+d=d^2$ vertices and contains exactly one terminal of type~$2$.
Let $\cB_{d}$ be the set of all composite-blocks.

A {\it backbone} of a component $Z$ of $F'$ is a minimum tree $T_Z$ contained in $Z$ that connects all terminals of $F'$ that belong to $Z$; observe that all leaves of $T_Z$ are terminals. 
The {\it skeleton}~$S$ of $F'$ is the graph obtained from $F'$ by removing all vertices from the blocks in $\cB_{l_1}$ except  terminals of type~1 and all vertices from the
composite-blocks 
in $\cB_{d}$ except terminals of type~2.
Note that every backbone is a subgraph of $S$. 
Furthermore, the vertices of the blocks in $\cB_{l_3}$ all belong to $S$.

\subsection{The Algorithm}\label{s-alg}

{\bf Outline.} Our polynomial-time algorithm consists of the following two phases:\\[-10pt]
\begin{enumerate}
\item {\it Branching Phase}, which consists of the following three steps:
\begin{itemize}
\item [{\bf 1.}] guessing the terminals of $F'$; 
\item [{\bf 2.}] guessing the backbones of the components of $F'$; and
\item [{\bf 3.}] guessing the skeleton of $F'$, and\\[-10pt]
\end{itemize}
\item {\it Completion Phase}, where we extend the partial solutions obtained in the Branching Phase to complete ones by finding non-skeleton vertices of $F'$ and trivial components of $F$; we do this by:
\begin{itemize}
\item [{\bf 1.}] reducing the problem to {\sc Max Weight Independent Set} for $sP_3$-free graphs using the blob graph construction in Section~\ref{s-pre}, and 
\item [{\bf 2.}] solving this problem using the polynomial-time algorithm of Brandst{\"{a}}dt and Mosca~\cite{BM18}.\\[-10pt]
\end{itemize}
\end{enumerate}
We now describe our algorithm, prove its correctness and perform a running time analysis. 

\medskip
\noindent
{\bf Branching Phase.}   
This phase of our algorithm consists of a series of guesses, where we find certain vertices and substructures in $G$.
The total number of vertices to be guessed will be $\Oh(s^2d^2)$. Since we guess them exhaustively, this results in a recursion tree with $\Oh(n^{\Oh(s^2 d^2)})$ leaves. As both $s$ and $d$ are constants, this bound is polynomial in $n$.
We will ensure that the optimum solution $F=G[X]$ will be found in the call corresponding to at least one of the leaves of the recursion tree.
Based on the properties of $F$, we will expect the guessed vertices to satisfy certain conditions.
If, at some point, the guessed vertices do not satisfy these conditions, we just terminate the current call, as it will not lead us to find $F$. This will be applied implicitly throughout the execution of the algorithm.

The branching phase is illustrated in Figures~\ref{fig:guess1}--\ref{fig:guess3}.
We use the convention that gray/black elements are still unknown and blue elements are the ones that we have already guessed.

\medskip
\noindent
{\bf Step 1. Guessing the terminals of $\mathbf{F'}$.}
We guess the set $C\subseteq V$ of terminals of $F'$. By Lemma~\ref{t-3}, the total number of terminals is bounded by $(2d+1)\cdot (s-1) \leq 3ds$. 
Hence, we consider $O(n^{3ds})$ options and for each chosen set $C$ we do as follows,
For each terminal in~$C$, we guess its type (1, 2, or both). This results in $3^{|C|} \leq 3^{3ds}$ possibilities.
We also guess the partition of $C$, corresponding to the connected components of $F$. This results in at most $|C|^{|C|} \leq (3ds)^{3ds}$ additional branches.
In total, we have $\Oh(n^{\Oh(ds)})$ branches. 
Step~1 is illustrated in Fig.~\ref{fig:guess1}.

\begin{figure}[t]
\centering
\includegraphics[scale=1.1,page=2]{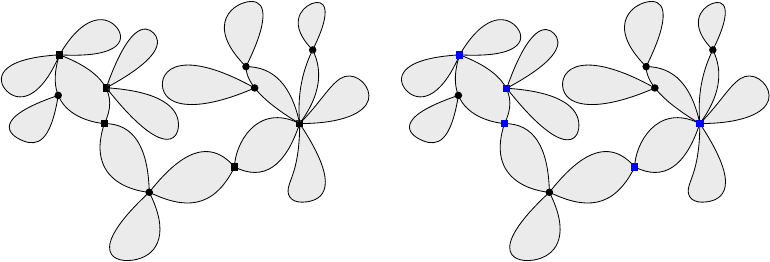}
\caption{Step 1 of the Branching Phase. Left: the graph $F'$. 
Right: the terminals of $F'$.}
\label{fig:guess1} 
\end{figure}

\medskip
\noindent
{\bf Step 2. Guessing the backbone of each component of $\mathbf{F'}$.}
Let $Z$ be a component of $F'$.
Let $C_Z \subseteq C$ be the subset of terminals that are in $Z$.
Let $T_Z$ be a backbone of $Z$.
Let $T_Z'$ be the tree obtained from $T_Z$ by contracting 
every path in $T_Z$ whose internal vertices are all non-terminals and of degree~$2$ to an edge.
Note that every non-terminal vertex of $T'_Z$ has degree at least~$3$.
Since $T'_Z$ has at most $|C_Z|$ vertices of degree at most~$2$, by the handshaking lemma we observe that the total number of vertices of $T'_Z$ is at most $2|C_Z|$.
Recall that every edge of $T'_Z$ corresponds to an induced path in $T_Z$.
Since $F'$ is 
$sP_3$-free 
and thus $P_{4s-1}$-free, we conclude that $T_Z$ has at most $2|C_Z| \cdot (4s-2) \leq 8s \cdot |C_Z|$ vertices.

\medskip
\noindent
Let $T$ be the forest whose components are 
the guessed backbones of the components of $F'$.
Note that the total number of vertices of $T$ is at most $\sum_{Z} 8s \cdot |C_Z| = 8s \cdot |C| \leq 24ds^2$.
Thus we may guess the whole forest $T$, which results in $\Oh(n^{\Oh(ds^2)})$ branches.
Step~2 is illustrated in Fig.~\ref{fig:guess2}.

\begin{figure}[t]
\centering
\includegraphics[scale=1.1,page=3]{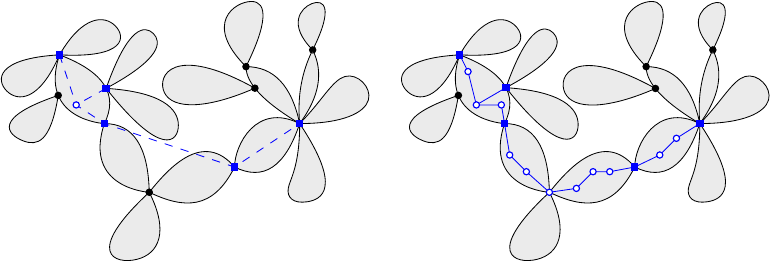}
\caption{Step 2 of the Branching Phase. Left: the tree $T'_Z$. Right: the tree $T_Z$.}
\label{fig:guess2} 
\end{figure}

\medskip
\noindent
{\bf Step 3. Guessing the skeleton of $\mathbf{F'}$.}
Let $T$ be the forest guessed in the previous step; recall that $T$ has at most $24ds^2$ vertices.
We guess the partition of $E(T)$ corresponding to \emph{blocks} of $F'$;
note that a vertex $v$ may be in several blocks: this happens precisely  if $v$ is a cutvertex in $F'$.
This results in at most $|E(T)|^{\Oh(|E(T)|)} \leq |V(T)|^{\Oh(|V(T)|)} \leq (ds)^{\Oh(ds^2)}$ branches.

We now discuss some properties of the 
(composite-)blocks.
We use the names of vertices as in the definitions introduced in Section~\ref{s-str}, recall also Fig.~\ref{fig:blocktypes}.
The crucial observation is that now there is a branch, where:\\[-10pt]
\begin{itemize}
\item For each block in $\cB_{l_1}$, we have guessed its cutvertex and no other vertices.
\item For each block in $\cB_{l_2}$, we have not guessed any vertices.
\item For each block in $\cB_{l_3}$, we guessed no vertices yet except possibly its cutvertex in $F'$ (but in the latter case we have not indicated this yet).
\item For each block in $\cB_{w}$, we have guessed
its type-2 terminal vertex
and we guessed no other vertices.
Thus, for each composite-block in $\cB_{d}$, we have guessed its cutvertex in~$F'$ and no other vertices.
\item For each block in $\cB_{in}$, we have guessed at least two vertices.\\[-10pt]
\end{itemize}
Now we proceed to the final guessing step, see Fig.~\ref{fig:guess3}.
First we guess all blocks in $\cB_{in}$. We can do it as\\[-10pt]
\begin{itemize}
\item (i) we know at least two vertices of such a block,
\item (ii) the number of these blocks is at most $|E(T)| \leq 24ds^2$, and
\item (iii) each block has at most $d$ vertices.\\[-10pt]
\end{itemize}
This results in at most $n^{\Oh(|E(T)| \cdot d)} = n^{\Oh(d^2s^2)}$ further branches. In each branch, we guessed all vertices of a skeleton~$S$ except those that are non-cutvertices of $F'$ that belong to the blocks in $\cB_{l_3}$. We will now guess which vertices of $S$ will also belong to exactly one block in $\cB_{l_3}$ (so these vertices will be cutvertices in $F'$). These vertices belong to at least one block of $\cB_{in}$. As the union of the vertices of the blocks in $\cB_{in}$  has size at most $24d^2s^2$, this leads to $\Oh(1)$ extra branches.

Finally, we guess all blocks in $\cB_{l_3}$. Note that we can do it, as\\[-10pt]
\begin{itemize}
\item (i) we know their cutvertices,
\item (ii) the number of these cutvertices is at most $24ds^2$, 
\item (iii) each cutvertex is contained in exactly one block from  $\cB_{l_3}$, and 
\item (iv) each block has at most $d$ vertices.\\[-10pt]
\end{itemize}
This results in at most $n^{\Oh(|V(T)| \cdot d)} = n^{\Oh(d^2s^2)}$ branches, that is, at most $n^{\Oh(d^2s^2)}$ sets that are potential skeletons $S$ of $F'$.
Step~3 is illustrated in Fig.~\ref{fig:guess3}.

\begin{figure}[t]
\centering
\includegraphics[scale=1.1,page=4]{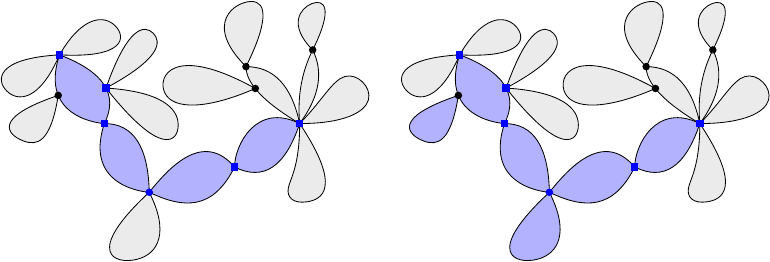}
\caption{Step 3 of the Branching Phase.
Left: our knowledge about $F'$ after guessing the blocks in $\cB_{in}$.
Right: our knowledge about $F'$ after guessing the blocks in $\cB_{l_3}$.}
\label{fig:guess3} 
\end{figure}

The following claim summarizes the outcome of the guessing phase of the algorithm.
\begin{claimm}\label{clm:guessing}
In time $\Oh(n^{\Oh(s^2d^2)})$ we can enumerate a collection $\cS$ of $\Oh(n^{\Oh(s^2d^2)})$ triples $(S, C_1, C_2)$,
where $S \subseteq V$ and $C_1,C_2 \subseteq S$ such that $\cS$ has the following property.
Let $X \subseteq V$, such that $F = G[X]$ is a $\cC$-block graph.
Let $X' \subseteq X$ be the vertex set of the graph $F'$ obtained from $F$ by removing all trivial components.
Then there is at least one triple $(S,C_1,C_2) \in \cS$, where 
\begin{enumerate}[a)]
\item $C_1$ is the set of terminals of type~1 in $F'$,
\item $C_2$ is the set of terminals of type~2 in $F'$,
\item $G[S]$ is the skeleton of $F'$. 
\end{enumerate}
\end{claimm}

\noindent
{\bf Completion Phase.}\label{sec:blobs}
Let $\cS$ be the the collection from Claim~\ref{clm:guessing} and let  $(S,C_1,C_2) \in \cS$ be 
a triple that satisfies the properties listed in the statement of Claim~\ref{clm:guessing} for an optimum solution $F = G[X]$.
Let $\cX :=  \cX_0 \cup \cX_1 \cup \cX_2$ be the family of subsets  of $V$ with:
\begin{align*}
\cX_0 := &  \{ \{v\} ~|~ v \in V \},\\
\cX_1 := & \{ B \subseteq V ~|~ G[B] \in \cC \}, \mbox{and}\\
\cX_2 := & \{ B \subseteq V ~|~ B~\text{is a 
composite-block
whose blocks are in } \cC \}.
\end{align*}
Let $G^{\cC}$ be the graph whose vertex set is $\cX$, and edges join sets that are adjacent in $G$.
Furthermore, we define a weight function
$\wei^{\cC} \colon \cX \to \Q^+$ as
\[\wei^{\cC}(A) = \sum_{v \in A} \wei(v).\]
Note that in order to complete $S$ to the optimum solution $F = G[X]$, we need to determine:\\[-10pt]
\begin{itemize}
\item all blocks in $\cB_{l_1}$,
\item all composite-blocks in $\cB_{d}$,
\item all trivial components of $F$.\\[-10pt]
\end{itemize}
\noindent
Note that the vertex sets of all these subgraphs are in the family $\cX$ and they form an independent set in  $G^{\cC}$.
Furthermore, since $X$ is of maximum weight, the total weight of selected subsets must be maximized.
Thus the idea behind the last step is to reduce the problem to solving \textsc{Max Weight Independent Set} in an appropriately defined subgraph of $G^{\cC}$ and weights $\wei^{\cC}$.

To ensure that the selected subsets are consistent with our guess $(S,C_1,C_2) \in \cS$,
we will remove certain vertices from $G^{\cC}$.
In particular, let $\cX'$ consist of the sets $A \in \cX$, such that:\\[-10pt]
\begin{enumerate}
\item $A \in \cX_0 \cup \cX_1$ and $A$ is non-adjacent to $S$; these are  the candidates for trivial components of $F$,
\item $A \in \cX_1$ and $A$ intersects $S$ in exactly one vertex, which is in $C_1$; these are the candidates for blocks in $\cB_{l_1}$,
\item $A \in \cX_2$ and $A$ intersects $S$ in exactly one vertex, which is in $C_2$ and is not the cutvertex of $G[A]$; these are the candidates for 
composite-blocks
in $\cB_{d}$.\\[-10pt]
\end{enumerate}
\noindent
Now let $\cI \subseteq \cX'$ be an independent set of $G^{\cC}$, and let $S' = \bigcup_{A \in \cI} A$.
It is straightforward to verify that if $(S,C_1,C_2) \in \cS$ satisfies the properties listed in Claim~\ref{clm:guessing}, then $G[S \cup S']$ is a $\cC$-block graph. 
Thus, in one of the branches, we will find the optimum solution $F = G[X]$.

Now let us argue that the last step can be performed in polynomial time.
First, observe that $|\cX| \leq n + n^{d} + n^{d^2} = n^{\Oh(d^2)}$ and the family $\cX$ can be exhaustively enumerated in time $n^{\Oh(d^2)}$.
Next, $\cX'$ can be computed in time polynomial in $|\cX|$, and thus in $n$.
This implies that the graph $G^{\cC}[\cX']$ can be computed in time polynomial in $n$.
We observe that $G^{\cC}$, and thus $G^{\cC}[\cX']$, is an induced subgraph of the blob graph $\Blob{G}$, introduced in Section~\ref{s-pre}. Hence, by \cref{thm:blob}, we conclude that $G^{\cC}[\cX']$ is $sP_3$-free.

The final ingredient is the polynomial-time algorithm for \textsc{Max Weight Independent Set} in $sP_3$-free graphs by Brandst{\"{a}}dt and Mosca~\cite{BM18}. Its running time on an $n'$-vertex graph is $n'^{\Oh(s)}$.
Since the number of vertices of $G^{\cC}[\cX']$ is 
$n^{\Oh(d^2)}$, we conclude that a maximum-weight independent set in $G^{\cC}[\cX']$ can be found in time $n^{\Oh(sd^2)}$.

Summing up, in the guessing phase, in time  $n^{\Oh(s^2d^2)}$ we enumerate the family $\cS$ of size $n^{\Oh(s^2d^2)}$.
Then, for each member $(S,C_1,C_2)$ of $\cS$, we try to extend the partial solution to a complete one. 
This takes time $n^{\Oh(sd^2)}$ per element of $\cS$. 
Among all found solutions, we return the one with maximum weight.
The total running time of the algorithm is $n^{\Oh(s^2d^2)}$, which is polynomial in $n$, since $s$ and $d$ are constants.
This completes the proof of Theorem~\ref{thm:main}.

\section{The Proof of Theorem~\ref{thm:P5sP1}}\label{s-mmm}

In this section we prove that for every integer $s\geq 1$ and every finite class $\cC$ of biconnected graphs,  \CBG can be solved in polynomial time for $(sP_1+P_5)$-free graphs.
In Section~\ref{s-boundary} we consider two boundary cases, namely the case where $\cC=\emptyset$ and the case where $s=0$. We will use these cases in our algorithm in Section~\ref{s-thealgo} after first proving some structural lemmas in Section~\ref{s-s2}.

\subsection{Two Boundary Cases}\label{s-boundary}

First assume that $\cC = \emptyset$. Recall that {\sc Max $\emptyset$-Block Graph} is equivalent to \textsc{Max Independent Set}. The latter problem is polynomial-time solvable for $P_5$-free graphs (and even for $P_6$-free graphs~\cite{GKPP19}).

\begin{theorem}[\cite{LVV14}]\label{t-p5i}
{\sc Max Independent Set} can be solved in polynomial time for $P_5$-free graphs.
\end{theorem}

We also recall the aforementioned and well-known observation from Section~\ref{s-intro} on graphs that are nearly~$\pi$ for some graph property~$\pi$ (see, for example,~\cite{BH07}). As a special case, we find that  {\sc Max Independent Set} for $(P_1+P_5)$-free graphs is polynomial-time solvable if it is so for $P_5$-free graphs. Combining Theorem~\ref{t-p5i} with $s$ applications of this argument leads to the following (known) extension of Theorem~\ref{t-p5i}, which we will need as a lemma.

\begin{lemma}\label{lem:misp5sp1}
For every fixed $s$, {\sc Max Independent Set} can be solved in polynomial time in $(sP_1+P_5)$-free graphs.
\end{lemma}

Now we deal with the case where $s=0$. That is, we consider \CBG restricted to $P_5$-free graphs when $\cC$ is a finite class of biconnected graphs. For this case we will use
\emph{Monadic Second-Order Logic} ($\mathsf{MSO}_2$) over graphs, which consists of formulas with vertex variables, edge variables, vertex set variables, and edge set variables, quantifiers, and standard logic operators. We also have a predicate $\mathsf{inc}(v,e)$, indicating that the vertex $v$ belongs to the edge $e$.

Abrishami et al.~\cite[Theorems~5.3 and~7.3]{ACPRS21} proved the following result, even for the extension \emph{Counting Monadic Second-Order Logic} ($\mathsf{CMSO}_2$) of $\mathsf{MSO}_2$, which allows atomic formulas of the form $|X|\equiv p\bmod q$, where $X$ is a set variable and $0\leq p<q$ are integers (however, we do not need this extension for our purposes)
We refer the reader to~\cite{CE12} for further information on $\mathsf{MSO}_2$ logic on graphs.

\begin{theorem}[\cite{ACPRS21}]\label{t-mso}
For every fixed $\mathsf{CMSO}_2$ formula~$\Phi$ and every constant $t$, it is possible for a $P_5$-free graph $G$ with weight function $\wei: V(G) \to \Q^+$, to find in polynomial time 
a maximum-weight set $X \subseteq V(G)$, such that $G[X]$ is of treewidth at most $t$ and satisfies~$\Phi$.
\end{theorem}

\noindent
We use Theorem~\ref{t-mso} in our next lemma.

\begin{lemma}\label{l-p5}
For every finite class $\cC$ of biconnected graphs, \CBG can be solved in polynomial time for $P_5$-free graphs.
\end{lemma}

\begin{proof}
Note that every $\cC$-block graph has treewidth at most $\max_{C \in \cC } |V(C)|$, which is a constant. In order to use Theorem~\ref{t-mso} it remains to show that the property that a set $X \subseteq V(G)$ induces a $\cC$-block graph in $G$ is expressible in $\mathsf{CMSO}_2$. 
We show that we can express this property already in $\mathsf{MSO}_2$.

In what follows, $x$ and $y$ are vertex variables, $e$ is an edge variable, while $X, X'$, and $Y$ denote vertex set variables.
We will use some standard shortcuts (see also~\cite{CE12}), for instance: 
\begin{align*}
\forall (x \in X) : \phi & \qquad \text{ stands for } \qquad \forall x : (x \in X) \Rightarrow \phi \qquad \text{ and}\\
\forall (X' \subseteq X) : \phi & \qquad \text{ stands for } \qquad \forall X' : (\forall (x \in X') : x \in X) \Rightarrow \phi.
\end{align*}
We can now show the required claim.
First, we express the property that $G[X]$ is connected in $\mathsf{MSO}_2$, in the usual way:
\[
\mathsf{connected}(X) := \forall (X' \subseteq X) : \left( \exists (x \in X') \exists (y \in X) \exists e : y \notin X' \land\mathsf{inc}(x,e) \land \mathsf{inc}(y,e) \right).
\]
We now express the property that $G[X]$ is biconnected in $\mathsf{MSO}_2$:
\[
\mathsf{biconnected}(X) := \mathsf{connected}(X) \land \forall (x \in X) : \mathsf{connected}(X \setminus \{x\}).
\]
Now $G[X]$ is a block of $G[Y]$ if it is biconnected and maximal with this property:
\[
\mathsf{block}(X,Y) := \mathsf{biconnected}(X) \land \forall (y \in Y) : (y \notin X) \Rightarrow \lnot \mathsf{biconnected}(X \cup \{y\}).
\]
If $C$ is a fixed graph, then the property that $G[X]$ is isomorphic to $C$ can be easily hard-coded in a formula. We denote this predicate by $\mathsf{is}\text{-}C(X)$. This can be extended to checking whether $G[X] \in \cC$ (if $\cC$ is finite) by setting
\[
\mathsf{in}\text{-}\cC(X) := \bigvee_{C \in \cC} \mathsf{is}\text{-}C(X).
\]
Finally, $G[X]$ is a $\cC$-block graph if and only if $X$ satisfies
\begin{equation} \label{eq:csmo}
\mathsf{is}\text{-}\cC\text{-}\mathsf{block\text{-}graph}(X) := \forall (X' \subseteq X) : \mathsf{block}(X) \Rightarrow \mathsf{in}\text{-}\cC(X).
\end{equation}
This completes the proof of the lemma.
\end{proof}

\subsection{Structural Lemmas}\label{s-s2}

Let $s\geq 0$, and let $G$ be the $(sP_1+P_5)$-free instance graph with weight function $\wei$. Let $\cC$ be a finite class of biconnected graphs. Let $d$ be the maximum number of vertices of a graph in $\cC$. Similarly to Section~\ref{s-poly}, we will analyze the structure of an (unknown) maximum-weight solution.
Let $X\subseteq V(G)$ 
be such that $F = G[X]$ is a $\cC$-block graph. Again we consider the block-cut forest $\BlockTree{F}$ of $F$.
Recall that a leaf block is a block which is a leaf of $\BlockTree{F}$.

\begin{lemma} \label{lem:boundeddeg}
Every cutvertex of $F$ belongs to at most $s+1$ non-leaf blocks.
\end{lemma}

\begin{proof}
For contradiction, let $x$ be a cutvertex that belongs to $s+2$ blocks $b_1,b_2,\ldots,b_{s+2}$.
Consider one such block $b_i$ for $i \in [s+2]$. As $b_i$ is not a leaf block, there is a cutvertex $y_i \in V(b_i) \setminus \{x_i\}$
and a block $b_i' \neq b_i$ containing $y_i$, see Fig.~\ref{fig:boundeddeg}. Note that $x \notin V(b'_i)$.
Let $y_i'$ be any vertex from $V(b_i') \setminus \{y_i\}$; note that $y_i'$ is non-adjacent to $x$.
Let $Q_i$ be a shortest $x$-$y_i'$-path contained in $V(b_i) \cup V(b_i')$ and note that $Q_i$ has at least two edges.
Furthermore, for $i,j \in [s+2]$, such that $i \neq j$, the paths $Q_i$ and $Q_j$ share one endvertex (namely $x$) and no other vertices.
Thus $G[V(Q_1) \cup V(Q_2)]$ is an induced path with at least five vertices and consequently, $G[V(Q_1) \cup V(Q_2) \cup \bigcup_{i=3}^{s+2} \{y'_i\}]$ contains an induced $sP_1+P_5$, a contradiction.
\end{proof}

\begin{figure}[htb]
\centering
\includegraphics[scale=1.3,page=7]{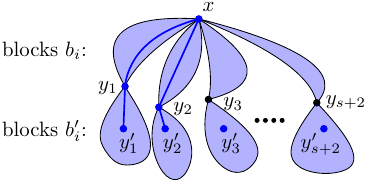}
\caption{The induced $sP_1+P_5$ in the proof of Lemma~\ref{lem:boundeddeg}.}
\label{fig:boundeddeg} 
\end{figure}

A vertex $x \in V(F)$ is called \emph{internal} if it is a cutvertex or belongs to a non-leaf block.
All other vertices are {\it external}, see Fig.~\ref{fig:internal}.

\begin{figure}[htb]
\centering
\includegraphics[scale=1.3,page=8]{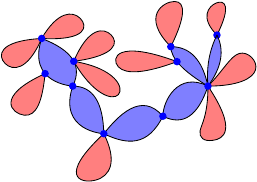}
\caption{Internal (blue) and external (red) vertices of $F$.}
\label{fig:internal} 
\end{figure}

\begin{lemma} \label{lem:boundedsize}
Every component of $F$ has at most  $(5+2s)(d(s+1))^{5+2s}$ internal vertices.
\end{lemma}
\begin{proof}
Let $X_{int}$ be the set of internal vertices of some component of $F$.
Each $v \in X_{int}$ is in at most $s+1$ blocks of $G[X_{int}]$ by Lemma~\ref{lem:boundeddeg}, 
and moreover, it has degree at most $d-1$ in each block (as each block has at most $d$ vertices). Thus the maximum degree in $G[X_{int}]$ is at most $(d-1)(s+1) \leq d(s+1)$.
As $G$ is $(sP_1+P_5)$-free, $G[X_{int}]$ is $(sP_1+P_5)$-free. Hence, $G[X_{int}]$ is also $P_{5+2s}$-free and as $G[X_{int}]$ is connected, it has diameter at most $5+2s-1$.
Every graph with maximum degree at most $d(s+1)$ and diameter at most $5+2s-1$ has at most
\[
1 + d(s+1) + (d(s+1))^2 + \ldots + (d(s+1))^{5+2s-1} \leq (5+2s)(d(s+1))^{5+2s}
\]
vertices\footnote{An astute reader might notice that this bound can actually be improved to the so-called Moore bound. However, as we do not try to optimize the constants, we kept bounds as simple as possible.}. This completes the proof of the lemma.
\end{proof}

We say that a component $F'$ of $F$ is \emph{big} if $|V(F')| \geq (ds+1) \cdot (5+2s)(d(s+1))^{5+2s}$. Otherwise $F'$ is \emph{small}.

\begin{lemma} \label{lem:bigvertex}
If a component of $F$ is big, then it has a cutvertex belonging to at least $s$ leaf blocks.
\end{lemma}
\begin{proof}
Let $X'$ be such that $G[X']=F'$ is a 
big
component of $F=G[X]$.
Let $X'_{int}$ and $X'_{ext}$ be the sets of internal and external vertices of $X'$, respectively.
Note that $X' =X'_{int}\cup X'_{ext}$ and $X'_{int}\cap X'_{ext} = \emptyset$.
By Lemma~\ref{lem:boundedsize} we have that $|X'_{int}|\leq (5+2s)(d(s+1))^{5+2s}$.
Consequently,
\[\begin{array}{lcl}
|X'_{ext}| &= &|V(F')|-|X'_{int}|\\[4pt]
&\geq &(ds+1) \cdot (5+2s)(d(s+1))^{5+2s} - |X'_{int}|\\[4pt] &\geq &ds \cdot (5+2s)(d(s+1))^{5+2s}.\end{array}\]
As every block contains at most $d$ vertices, the above implies that $G[X']$ has at least $s \cdot (5+2s)(d(s+1))^{5+2s} \geq s \cdot |X'_{int}|$ leaf blocks.
Each leaf block contains exactly one internal vertex, so by the pigeonhole principle we conclude that there must be an internal vertex 
belonging to at least $s$ leaf blocks. This completes the proof of the lemma.
\end{proof}

\subsection{The Algorithm}\label{s-thealgo}

We are now ready to present our polynomial-time algorithm for $(sP_1+P_5)$-free graphs. Let $s\geq 0$, and let $G$ be the $(sP_1+P_5)$-free instance graph with weight function $\wei$.
Let $\cC$ be a finite class of biconnected graphs.
Let $d$ be the maximum number of vertices of a graph in $\cC$.
Recall that $F$ is the optimum solution we are looking for.

The algorithm consist of three phases, in each of which we look for solutions of a specific type.
Afterwards, the algorithm returns the maximum solution found during the whole execution.

\medskip
\noindent\textbf{Case 1: $F$ has at most three big components.}\\
First suppose that $F$ has exactly three big components $F^1=G[X^1]$, $F^2=G[X^2]$, and $F^3=G[X^3]$.  
See Fig.~\ref{fig:threebig}.
For $j \in [3]$, let $X^j_{int}$ be the set of internal vertices of $G[X^j]$ (depicted in blue in Fig.~\ref{fig:threebig}).
By Lemma~\ref{lem:boundedsize} we have that $|X^j_{int}| \leq (5+2s)(d(s+1))^{5+2s}$ and thus the set $X_{int} := \bigcup_{i \in [3]} X^j_{int}$ has at most $3(5+2s)(d(s+1))^{5+2s}$ vertices.
We guess the vertices from $X_{int}$ exhaustively; this results in $\Oh(n^{(5+2s)(d(s+1))^{5+2s}})$ branches.
We discard the branches where $G[X_{int}]$ is not a $\cC$-block graph with three components.

\begin{figure}[htb]
\centering
\includegraphics[scale=1,page=9]{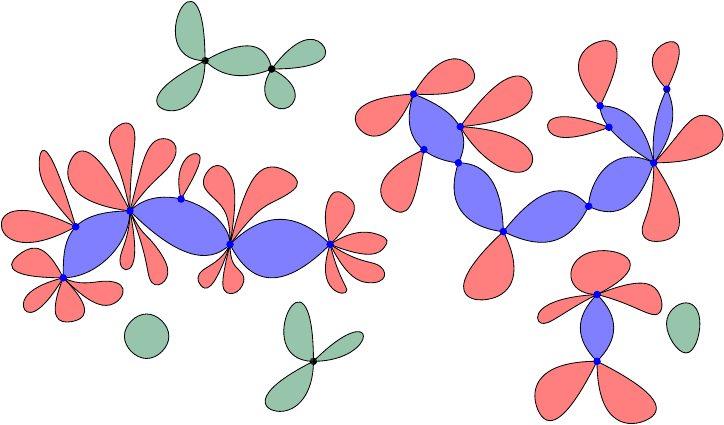}
\caption{Case 1. in the algorithm. Internal vertices of the three big components of $F$ are marked blue, while the external ones are red. Small components are marked green.}
\label{fig:threebig} 
\end{figure}

For each $X_{int}$ that we have not discarded the only thing left to do is to find:\\[-10pt]
\begin{itemize}
\item the small components of $F$ (marked green in Fig.~\ref{fig:threebig}) ,
\item the leaf blocks of $F^1$, $F^2$, and $F^3$ (marked red in Fig.~\ref{fig:threebig}).\\[-10pt]
\end{itemize}
\noindent
This task is very similar to the final case of the algorithm in Section~\ref{s-alg}. Let $\cX_s$ be the family of those subsets of $V(G) \setminus N[X_{int}]$ of size smaller than $(ds+1) \cdot (5+2s)(d(s+1))^{5+2s}$ that induce $\cC$-block graphs. The elements of $\cX_s$ are potential candidates for the vertex sets of small components of $F$.
The family $\cX_s$ can be enumerated in time $\Oh(n^{(ds+1) \cdot (5+2s)(d(s+1))^{5+2s}})$.

Let $\cX_\ell$ be the family of the sets $S \subseteq V(G)\setminus X_{int}$, satisfying the following properties:\\[-10pt]
\begin{enumerate}[(a)]
\item there is a unique $x \in X_{int}$ with a neighbour in $S$,
\item $G[S \cup \{x\}]$ is a graph from $\cC$.\\[-10pt]
\end{enumerate}
\noindent
The elements $S$ of $\cX_\ell$ are potential candidates for the sets of external vertices in the leaf blocks of $F^1$, $F^2$, and $F^3$, where $x$ is the unique neighbour of the block $S \cup \{x\}$ in $\BlockTree{F}$.
As each block has at most $d$ vertices, each set from $\cX_\ell$ has at most $d-1$ vertices. Hence, the family $\cX_\ell$ can be enumerated in time $\Oh(n^{d-1})$.

We now define $\cX:= \cX_S \cup \cX_\ell$ and have reduced to {\sc Max Independent Set} for $(sP_1+P_5)$-free graphs. Namely, we build in polynomial time the induced subgraph $\Blob{G}[\cX]$ of $\Blob{G}$ and the task is to find a maximum independent set in $\Blob{G}[\cX]$. As $\Blob{G}[\cX]$ is $(sP_1+P_5)$-free by Theorem~\ref{thm:blob}, we can use the polynomial-time algorithm from Lemma~\ref{lem:misp5sp1} for doing this. Afterwards, we use the solution found, together with $\cX$, to construct a forest $F$ for $G$. Out of all the forests found in this way, we remember one with maximum weight.

The algorithm also considers the three subcases where $F$ has zero, one, or two big components along the same lines as above but with some straightforward adjustments. In the end it returns a maximum-weight solution amongst the four solutions found. The total running time of Case~1 is polynomial, as there are 
$\Oh(n^{(5+2s)(d(s+1))^{5+2s})})$ branches and each of them is processed in time $n^{\Oh(d-1)}$, i.e., polynomial in $n$.

\medskip
\noindent\textbf{Case 2: $F$ has at least four big components.}\\
Let $X^1,X^2,X^3,X^4$ be the vertex sets of pairwise distinct big components of $F=G[X]$.
For each $j\in [4]$, there is $x_j \in X^j$ that belongs to at least $s$ leaf blocks of $F$ by Lemma~\ref{lem:bigvertex} .
Choose $s$ leaf blocks $b_1^j,\ldots,b^s_j$ containing $x^j$ and let $L^j := (\bigcup_{i=1}^s V(b^j_i)) \setminus \{x_j\}$.
Let $L := \bigcup_{j \in [4]} L^j$ and let $G' := G - (N[L] \setminus \{x_1,x_2,x_3,x_4\})$.

Now consider any $X' \subseteq V(G')$, such that $G'[X']$ is a $\cC$-block graph.
We observe that $G[X'\cup L]$ is also a $\cC$-block graph.

Due to the above observation we can proceed as follows. We will guess $x_1,x_2,x_3,x_4$, and $L$ exhaustively. 
Note that in the intended solution $\{x_1,x_2,x_3,x_4\} \cup L$ should be a $\cC$-block graph whose block-cute forest is a disjoint union of four starts with $x_1,x_2,x_3,x_4$ as centers. If this is not the case for some guess, we discard the branch.
As $|L| \leq 4ds$, the number of branches is $\Oh(n^{4+4ds})$. In each of those that we did not discard we will consider the graph $G'=G - (N[L] \setminus \{x_1,x_2,x_3,x_4\})$ and find a maximum-weight set $X' \subseteq V(G')$ such that $G'[X']$ is a $\cC$-block graph.
Then, by the above observation, $X' \cup L$ induces a $\cC$-block graph in $G$.
We will return the maximum-weight solution among all found in the branches.

The only thing left is to show that \CBG can be solved in polynomial time for $G'$.
For this, we make the following combinatorial claim.

\begin{lemma}\label{l-ppp}
The graph $G'$ is $P_5$-free.
\end{lemma}
\begin{proof}
For contradiction, suppose that $G'$ contains an induced $P_5$.
As $\{x_1,x_2,x_3,x_4\}$ is an independent set, at least one vertex from this set, say $x_1$, does not belong to this path.
Thus there exists an induced $P_5$ in $G' - x_1$.

Note that $G[L^1]$ contains an independent set $I$ of size $s$: it is sufficient to take one vertex from each block.
Furthermore, no vertex from $N[I]$ is in $G'-x_1$. Consequently, the induced~$P_5$ in $G'-x_1$, together with $I$, 
induces an $sP_1+P_5$ in $G$, a contradiction.
\end{proof}

\noindent
Due to Lemma~\ref{l-ppp}, \CBG in $G'$ can be solved in polynomial time by 
Lemma~\ref{l-p5}.\footnote{Both the bound on the treewidth of a $\cC$-block graph and the formula \eqref{eq:csmo} depend on $d$, and the dependence of these parameters (especially the $\mathsf{CMSO}_2$ formula) in the work of Abrishami et al.~\cite{ACPRS21} is quite involved.
Nevertheless, as $d$ is a constant, the running time of the algorithm in Lemma~\ref{l-p5} is polynomial.}

The total running time of Case~2 is polynomial, as there are $\Oh(n^{4+4ds})$ branches and processing each branch takes polynomial time. This completes the proof of Theorem~\ref{thm:P5sP1}.

\section{Hardness Results for Even Cycle Transversal on $\mathbf{H}$-Free Graphs}\label{s-ect}

In this section we prove that subject to a number of unsolved cases, the complexity of {\sc Even Cycle Transversal} for $H$-free graphs coincides with the one for {\sc Feedback Vertex Set.} 

An {\it odd cycle factor} of a graph $G$ is a set of odd cycles such that every vertex of $G$ belongs to exactly one of them. 
The {\sc Odd Cycle Factor} problem, which asks if a graph has an odd cycle factor, is known to be \NP-complete~\cite{Vo79}. 
The {\it line graph} $L(G)$ of a graph $G=(V,E)$ has vertex set $E$ and an edge between two distinct vertices $e$ and $f$ if and only if $e$ and $f$ share an end-vertex in $G$.
 
The proof of our next result for line graphs is somewhat similar to a proof for {\sc Odd Cycle Transversal} of~\cite{CHJMP18} but uses some different arguments as well.

\begin{theorem}\label{t-line}
{\sc Even Cycle Transversal} is \NP-complete for line graphs.
\end{theorem}

\begin{proof}
Let $G=(V,E)$ be an instance of {\sc Odd Cycle Factor} with $n$ vertices and $m$ edges.
We claim that $G$ has an odd cycle factor if and only if its line graph $L :=  L(G)$ has an even cycle transversal of size at most $m-n$, see Fig.~\ref{fig:ocf-ect}.

First suppose $G$ has an odd cycle factor.
Then there is $E' \subseteq E$, such that $|E'|=n$ and $L[E']$ is a disjoint union of odd cycles.
Hence, $S := E\setminus E'$ is an even cycle transversal of $L$ of size $|E|-n=m-n$. 
Now suppose $L$ has an even cycle transversal~$S$ with $|S| \leq m-n$.
Let $E':=E \setminus S$, As $|E|=m$, we have $|E'| \geq n$.

We prove the following claim.

\begin{claimm}\label{clm:linegraphs}
Every component of $L[E']$ is either an odd cycle or the line graph of a tree.
\end{claimm}

\begin{claimproof}
Let $D$ be a component of $L[E']$. 
If $D$ has no cycle, then $D$ is a path, as $L$ is a line graph and thus is claw-free.
Hence, $D$ is the line graph of a path, and thus a tree.

So suppose $D$ has a cycle $C$. Then $C$ is odd and induced, as $L[E']$ is an odd cactus.
If $D$ has no vertices except for the ones of $C$, then $D$ is an odd cycle and we are done.
Suppose otherwise.

First, assume that $C$ has at least five vertices.
Since $D$ has vertices outside $C$, there is a vertex of $C$ with a neighbour outside $C$.
Hence, $D$ contains either an even cycle or an induced claw, both of which are not possible.
So now suppose that $C$ has at most four vertices. Then $C$ is a triangle, as $D$ has no even cycles.
Since $D$ is an induced subgraph of $L$, there exists a subgraph $T$ of $G$ such that $D=L(T)$.
As $D$ is a connected graph with at least four vertices, containing a triangle, $T$ is a connected graph with at least four vertices.

We aim to show that $T$ is a tree.
For contradiction, suppose that $T$ contains a cycle $C_T$.
Then $C_T$ must be a triangle, as otherwise $D$ would contain an even cycle or an odd cycle with at least five vertices.
Let $a,b,c$ be the vertices of $C_T$. As $T$ is connected and has at least four vertices, at least one of $\{a,b,c\}$, say $a$, must have a neighbour $d \notin \{b,c\}$. However, the edges $ad-ab-bc-ac$ form a $C_4$ in $D$, a contradiction with $D$ being an odd cactus.
So we conclude that $T$ contains no cycles and thus $T$ is a tree.
\end{claimproof}

Each component of $L[E']$ that is an odd cycle corresponds to an odd cycle in $G$.
By Claim~\ref{clm:linegraphs}, each component $D$ of $L[E']$ that is not an odd cycle is the line graph of some subtree $T$ of $G$.
So, if $D$ has $r$ vertices, then $T$ has $r+1$ vertices.
Furthermore, the vertex sets of~$G$ corresponding to distinct components of $L[E']$ are pairwise disjoint.
Suppose that $L[E']$ has $p \geq 0$ components that are not odd cycles.
Let $Q$ be the set of vertices incident to at least one edge of $E'$.
Then $n = |V(G)| \geq |Q| = |E'|+ p \geq n+p$.
Hence, $p=0$ and $|Q|=n$. So, the components of $L[E']$ correspond to an odd cycle factor of $G$. This completes the proof.
\end{proof}

\begin{figure}[t]
\centering
\includegraphics[scale=0.9,page=6]{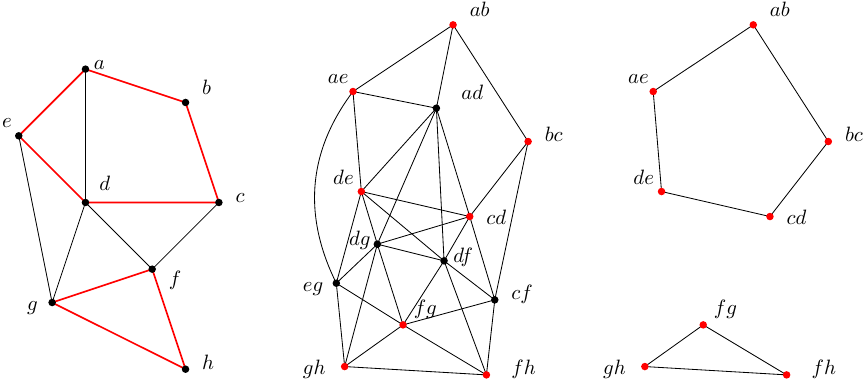}
\caption{Left: a graph $G$ with an odd cycle factor. Middle: the graph $L=L(G)$ and the set $E'$ (red). Black vertices form an even cycle factor. Right: the odd cactus $L[E']$.}
\label{fig:ocf-ect} 
\end{figure}

We make a straightforward observation similar to an observation for {\sc Feedback Vertex Set}~\cite{CHJMP18,Po74}, except that we must subdivide edges of a graph an even number of times.

\begin{theorem}\label{t-girth}
For every $p\geq 3$, {\sc Even Cycle Transversal} is \NP-complete for graphs of girth at least $p$.
\end{theorem}

\begin{proof}
We reduce from {\sc Even Cycle Transversal} for general graphs by noting the following. 
Namely, the size of a minimum even cycle transversal in $G$ is equal to the size of a minimum even cycle transversal in the graph $G'$ obtained from $G$ by subdividing every edge $2p$ times, and the girth of $G'$ is at least $p$.
\end{proof}

The next theorem is analogous to the one for {\sc Feedback Vertex Set}; see also Table~\ref{t-table}.

\begin{theorem}
Let $H$ be a graph. Then {\sc Even Cycle Transversal} for $H$-free graphs is polynomial-time solvable if $H\ssi sP_3$ or $H\ssi sP_1+P_5$ for some $s\geq 0$, and it is \NP-complete if $H$ is not a linear forest.
\end{theorem}

\begin{proof}
If $H\ssi sP_3$ or $H\ssi sP_1+P_5$ for some integer $s\geq 0$, then we use Corollary~\ref{c-1}.
If $H$ is not a linear forest, then it has a cycle or a claw. If $H$ has a cycle, then we apply Theorem~\ref{t-girth} for $p = |V(H)|+1$. Otherwise, $H$ has an induced claw and we apply Theorem~\ref{t-line}.
\end{proof}

\section{Conclusions}\label{s-con}

We proved that the \CBG problem is polynomial-time solvable on $sP_3$-free graphs and $(sP_1+P_5)$-free graphs (for every $s\geq 1$)
Hence, we have showed that for a large family of graphs ${\cal F}$, the {\sc Min ${\cal F}$-Transversal} problem is polynomial-time solvable on these graph classes.
The two best-known problems in this framework are {\sc Feedback Vertex Set} and {\sc Even Cycle Transversal}. Our results for {\sc Feedback Vertex Set} extend all the known polynomial-time results for {\sc Feedback Vertex Set} on $H$-free graphs, namely for $sP_2$-free graphs~\cite{CHJMP18}, $(sP_1+P_3)$-free graphs~\cite{DFJPPR19} and $P_5$-free graphs~\cite{ACPRS21}. By proving some new hardness results we also showed that in contrast to the situation for {\sc Odd Cycle Transversal}, all other known complexity results for {\sc Feedback Vertex Set} on $H$-free graphs hold for {\sc Even Cycle Transversal} as well. Hence, so far both problems behave the same on special graph classes. 

Due to the above, it would be interesting to prove polynomial equivalency of the two problems more generally. Table~\ref{t-table} still shows some missing cases for each of the three problems. 
In particular, we highlight the following borderline cases, namely the cases $H=P_2+P_4$ and $H=P_6$ for {\sc Feedback Vertex Set} and {\sc Even Cycle Transversal} and the case $H=P_1+P_4$ for {\sc Odd Cycle Transversal}.

We recall that in Section~\ref{s-boundary} we showed that the \CBG problem 
is a special case of finding a maximum-weight subset of vertices that induces a bounded-treewidth graph which satisfies a given $\mathsf{CMSO}_2$ formula. 
The latter problem can be solved in \emph{quasipolynomial time} for $P_r$-free graphs for any fixed $r$~\cite{GLPPR21}.
Thus we immediately obtain the following.

\begin{corollary}
For every linear forest $H$ and every \emph{finite} class $\cC$ of biconnected graphs, 
\CBG 
can be solved in quasipolynomial time for $H$-free graphs.
\end{corollary}

\noindent
In particular, this implies quasipolynomial-time algorithms for \textsc{Feedback Vertex Set} and {\sc Even Cycle Transversal} for $H$-free graphs if $H$ is a linear forest, whereas {\sc Odd Cycle Transversal} is \NP-complete even for $P_6$-free graphs~\cite{DFJPPR19}. Hence, a polynomial-time algorithm for {\sc Feedback Vertex Set} and {\sc Even Cycle Transversal} on $P_6$-free graphs would show that these two problems, restricted to $H$-free graphs, differ in their complexity from {\sc Odd Cycle Transversal}.

\medskip
\noindent
{\it Acknowledgements.} The first author thanks Carl Feghali for an inspiring initial discussion. The third author thanks Marcin Pilipczuk for some fruitful discussion including an alternative polynomial-time algorithm for {\sc Feedback Vertex Set} on $(P_1+P_5)$-free graphs.

\end{document}